\documentclass{article}[12pt]
\usepackage[letterpaper]{geometry}
\usepackage[english]{babel}
\usepackage{amsthm}
\usepackage{amsfonts}
\usepackage{graphicx}
\usepackage[style=numeric, 
            sorting=none,
            giveninits=true]{biblatex}
\usepackage[dvipsnames]{xcolor}
\usepackage{authblk}

\usepackage{chngcntr}
\counterwithout{figure}{section}

\newtheorem{proposition}{Proposition}

\title{Modeling Two-Scale Rank Distributions via Redistribution Dynamics or an Analytic Derivation of the Beta Rank Function}

\author[1]{Oscar Fontanelli\thanks{Corresponding author: oscar.fontanelli@flacso.edu.mx}}
\author[2,3]{Wentian Li}

\affil[1]{\textit{Facultad Latinoamericana de Ciencias Sociales México (FLACSO México), Mexico City, Mexico}}
\affil[2]{\textit{Department of Applied Mathematics and Statistics, Stony Brook University, Stony Brook, NY, USA}}
\affil[3]{\textit{The Robert S. Boas Center for Genomics and Human Genetics\\
The Feinstein Institutes for Medical Research, Northwell Health, Manhasset, NY, USA}}

\date{January 2026}

\AtBeginDocument{}

\begin{document}

\maketitle

\begin{abstract}
Beta Rank Function (BRF) is a two-sided distribution characterized by a smooth peak and double power-law decay, widely used to model empirical data exhibiting deviations from pure power laws. In this paper, we introduce a novel two-step generative process that produces data exactly following the BRF distribution. The first step involves any mechanism generating a power-law distribution, while the second step applies a regressive redistribution process that reallocates resources from poorer to richer entities, thereby amplifying inequality. This approach represents the first analytic derivation of an exact BRF distribution from a generative mechanism. We validate the model through applications to income and urban population distributions. Beyond exact generation, this framework offers new insights into the systemic origins of deviations from power laws frequently observed in complex systems, linking rank distributions to underlying feedback and redistribution dynamics.

\end{abstract}


\section{Introduction}

Long or heavy-tailed distributions \cite{anderson2006long}, such as power laws or stretched exponentials, are pervasive across natural, social and engineered systems, serving as essential tools for statistical data analysis in complex systems \cite{newman2005power, schroeder2009fractals}. Among these, the Discrete Generalized Beta Distribution (DGBD) and its associated Beta Rank Distribution (BRF) have emerged as highly flexible models that empirically capture deviations from pure power laws frequently observed in rank-size data \cite{gustavo}. These deviations, manifested as elbows or curvature in log-log rank-size plots, reflect a more complex underlying multi-scale dynamics and emergent behaviors \cite{clauset}.

Despite numerous empirical applications of the DGBD/BRF across domains — including city populations \cite{ghosh}, administrative unit populations \cite{oscar-au}, linguistic data \cite{wli-jql2, wli-jql, wli-entropy, wli-physica}, rank-citation profiles of scientists \cite{petersen}, centrality measures in mobility networks \cite{fontanelli2023intermunicipal}, log-returns of financial indexes \cite{fontanelli2020distribuciones}, among others — there has been a notable absence of theoretical mechanisms that generate an exact DGBD/BRF distribution. This gap limits the understanding of the systemic processes that give rise to such two-scale behaviors in complex systems.

In this work, we address this fundamental gap by proposing a novel two-step generative process that analytically produces the exact BRF distribution. The first step involves any mechanism generating a power law, a well-established signature of scale-free phenomena in complex systems \cite{mitzenmacher2004brief, sornette2006critical, sornette2009, lux2016financial}. The second step introduces a regressive redistribution process, reallocating resources asymmetrically from poorer to richer entities, thereby intensifying inequality and breaking the original scale invariance. This redistribution acts as a system-level feedback mechanism, modulating the initial scale-free distribution and producing emergent two-scale behavior, which is captured by the two-exponent BRF \cite{fontanelli2022beta}.

Methodologically, our approach uses quantile functions, which are the inverse of cumulative distribution functions, marking a departure from traditional probability density function or rank function analyses. By modulating the quantile function of a power law distribution, we achieve an exact analytical expression for the BRF, providing a rigorous mathematical framework to model deviations from power laws. This quantile-based methodology represents an innovation in complex systems modeling, enabling precise characterization of rank-based data and emergent distributional features.

A complex system has two characteristic scales when distinct behaviors ore dynamics occur at two different length or time scales. The two parameters of the BRF distribution effectively capture the behavior on these type of systems. By situating the BRF within this framework, our model offers new insights into inequality dynamics and systemic redistribution processes in complex adaptive systems. This work not only fills a theoretical void but also enhances the interpretability and applicability of the DGBD/BRF family as a versatile tool for analyzing ranked data across diverse complex systems.

The structure of this paper is the following: on section \ref{section_background_1} we review the known generating mechanisms for the BRF/DGBD; since the BRF distribution is defined through its quantile function, we give on section \ref{section_background_2} a quick review of quantile function modelling and give a formal definition of the BRF distribution; on section \ref{section_results_1} we state our main result, showing that our proposed transformation of a power law exactly produces a BRF; on section \ref{section_results_2} we show some numerical simulations to illustrate some relevant aspects of the transformation; on sections \ref{section_results_3} and \ref{section_results_4} we give empirical examples; finally, we discuss the implications of our model. 

\section{Background}

\subsection{Some previously proposed generating mechanisms for the DGBD rank function/BRF distribution}
\label{section_background_1}

The Discrete Generalized Beta Distribution (DGBD) rank function is defined as

\begin{equation}
x(r) = C \frac{(N+1-r)^b}{r^a},
\label{DGBD}
\end{equation}

\noindent where $x$ is the size of an observation, $r$ is the rank, $N$ the number of observations, $C$ a scale parameter, and $a, b \geq 0$ are shape parameters controlling decay rates on each side of the mode. The Beta Rank Function (BRF) is the associated probability distribution characterized by this rank function. The probability distribution function (pdf) of the BRF does not have an analytic expression \cite{fontanelli2022beta} except for some particular cases like $a=b$ \cite{fontanelli2016beyond}. Unlike the Pareto distribution, the BRF exhibits a two-sided power law behavior with distinct exponents on either side of its mode, making it highly flexible for modeling rank-size data with deviations from pure power laws.\\

Several mechanisms have been proposed to generate the DGBD/BRF, including:

\begin{enumerate}
\item Ranking of multinomial processes: complex systems composed of many subsystems produce macroscopic states whose rank distributions approximate the DGBD when the state space is large \cite{naumis2008tail}.

\item Expansion-modification algorithm: a dynamical system modeling duplication and mutation processes generates sequences whose rank-size distributions fit the DGBD, with parameters linked to persistence and change \cite{li1991expansion}.

\item Split-merge model: iterative splitting of the largest and merging of the smallest observations transforms an initial power law into a distribution well described by the DGBD, with shape parameters associated with merging and splitting dynamics \cite{li2016size}.

\item Bounded random subtraction: a stochastic process of iteratively subtracting random values under positivity constraints converges numerically to a heavy-tailed attractor approximated by the DGBD \cite{del2011general}.

\item Birth-death processes on networks: master equations modeling birth-death dynamics yield steady-state solutions fitting the DGBD rank function \cite{alvarez2014birth}.

\item Symbolic dynamics from unimodal mappings: length distributions of n-tuples generated by unimodal maps numerically follow the DGBD, with analytical support from thermodynamic formalism \cite{alvarez2018rank}.
\end{enumerate}

Despite these diverse approaches, none provide an exact analytical mechanism to generate the DGBD/BRF distribution. Addressing this gap, the present work proposes a two-step generative process that produces the exact BRF analytically, representing a theoretical advance in understanding the origins of two-scale behaviors in complex systems.

\subsection{Basic aspects of quantile function modeling}
\label{section_background_2}

Rank-size functions characterize a distribution by arranging empirical observations in descending order and expressing each value as a function of its rank. In contrast, quantile functions describe theoretical distributions by mapping probabilities to values, ensuring that the probability of the variable being less than or equal to a given value corresponds to the quantile. Rank-size plots effectively represent an empirical quantile function derived from a ranked sample, where rank aligns with empirical probability levels (but with opposite directions). Conversely, quantile functions serve as a continuous theoretical framework for describing distributional patterns. This section briefly introduces fundamental concepts of quantile functions.

The cumulative distribution function (cdf) $F_X(x)$ of a random variable $X$ is the total probability $p$ for that variable being smaller than a specific value $X=x$. The inverse of the cdf is called the quantile function $Q(p)$ and it gives the value $x$ given the total probability $p$ up to this value. If we think of the cdf as a function that maps $x \rightarrow p$, then the quantile function maps $p \rightarrow x$. The cut point $x$ is called the $p-$quantile of the distribution. To summarize, 

\begin{eqnarray}
cdf: & & F_X(x)= P[X \le x_p]=p \nonumber \\
quantile function: & & Q(p)= F_X^{-1}(p)=x_p. \nonumber
\end{eqnarray}

\noindent If $F$ is not strictly monotonic, the quantile function is defined as
$$
Q(p) = \inf \{x\in \mathbb{R} | p \leq F(x)\}.
$$

\noindent The discrete version of the quantile function is related, with negative direction, to the rank-size plot of the empirical data, similar to the probability density function (pdf) being related to the histogram of the empirical data.  Let $X_1,...,X_N$ by \emph{iid} random variables (a random sample from a certain distribution). By means of the empirical cumulative distribution function
$$
\displaystyle \hat{F}(x) = \frac{1}{N}\sum_{i=1}^N 1[X_i \leq x],
$$
\noindent we can define the \emph{rank} of an observation $x$ within a sample of size $N$ as
$$
r(x) = 1+N(1-\hat{F}(x)).
$$ 

\noindent This is consistent with the common practice, where $r(x_{max})=1$ and $r(x_{min})=N$. Note that, if the sample comes from a continuous distribution, the probability of ties is zero. The rank value can be converted to total probability $p$ by $p = (N+1-r)/N $. 
Note that cdf refers to the total probability for $X$ being smaller
than a specific value. When the largest value is observed, $p=1$, $X$ reaches its maximum value. But the largest value has rank=1, and number 1 is the smallest possible value.  Therefore, we say the rank variable $r$ and total probability $p$ have the opposite directions.

Fig. \ref{plot_examples} shows the pdf and quantile function for a lognormal distribution, as well as the histogram and rank-size plot of a sample from this distribution (we chose this distribution for clarity in the visualization). As expected, the histogram resembles the density function (it misses the shape near $x=0$ due to finite bin sizes), whereas empirical rank-size plot resembles the theoretical quantile function after reflection on the y-axis at $x=0$.  

\begin{figure}[tp]
\centering
\includegraphics[width=0.5\linewidth]{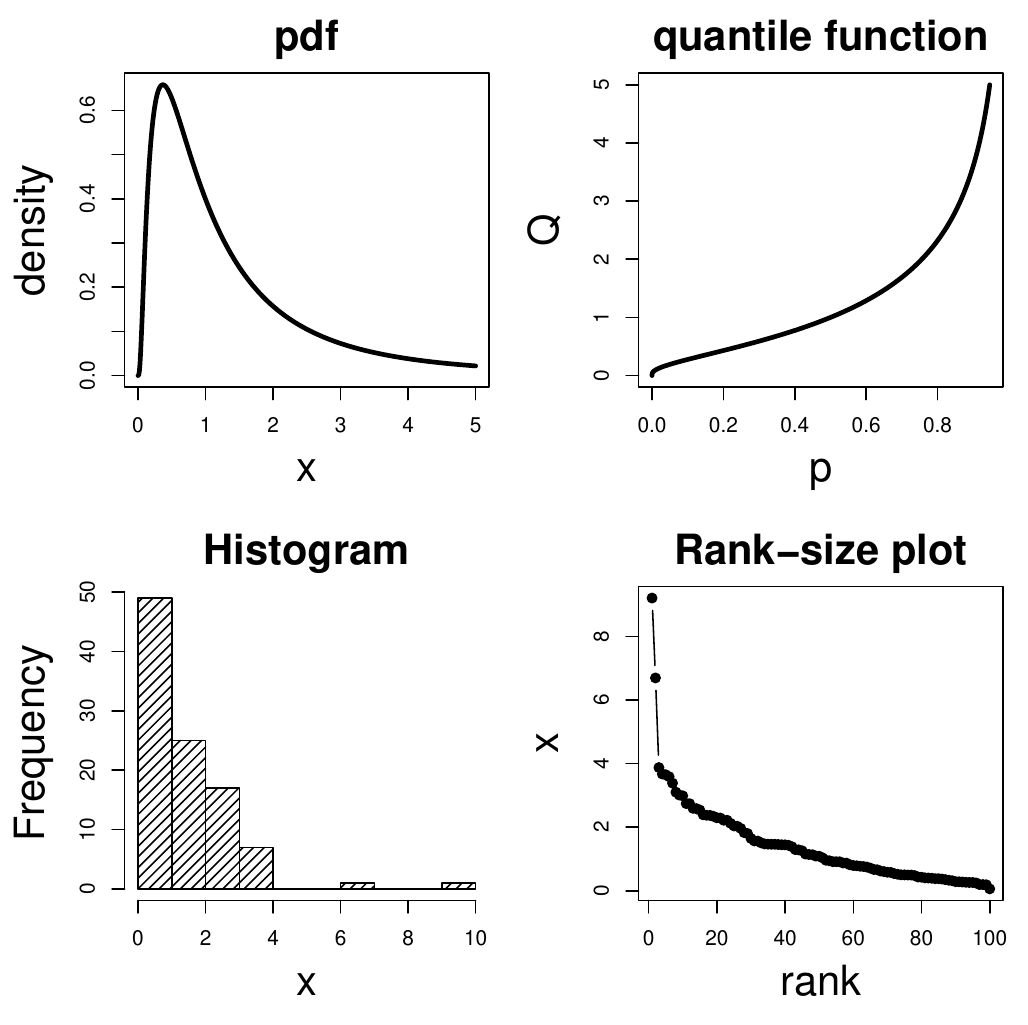}
\caption{\textbf{Quantile - rank size function correspondence.} Top panels show the probability density and the quantile function of a lognormal distribution with standard parameters. Bottom panels show the histogram and rank-size plot of a random sample (size=100) from the same distribution.}
\label{plot_examples}
\end{figure}

For the purposes of this work, we will need the quantile function of two distributions: the Pareto or power law and the BRF. If X is a random variable with a Pareto distribution ($X \sim Pareto(x_m, a)$),  the pdf, cdf, and quantile function are:

\begin{eqnarray}
pdf & & f_X(x) = \frac{(x_m)^{1/a}}{a} \frac{1}{x^{1/a+1}} 
\mbox{for x $\in [x_m,\infty$)} \nonumber \\
cdf: & & F_X(x) = 1- \left( \frac{x_m}{x}\right)^{1/a}. \nonumber \\
quantile function: & & Q_X(p)= \frac{x_m}{(1-p)^a} \nonumber
\end{eqnarray}

For BRF distribution with parameters $C$, $a$ and $b$

\begin{eqnarray}
pdf : & & 
\mbox{no analytic expression} \nonumber \\
cdf: & & \mbox{no analytic expression} \nonumber \\
quantile function: & & Q_{BRF}(p)= C \frac{p^b}{(1-p)^a},  \nonumber
\end{eqnarray}

\noindent where $C>0$ is a scale parameter, while $a\geq 0$ and $b\geq 0$ are shape parameters, each controlling the decay on one side of the mode.  This distribution is also known in the literature as the Davis distribution \cite{hankin} or the power-Pareto distribution \cite{gilchrist2000statistical}. 

In brief, quantile functions are used to describe the distribution of values within a dataset or system by focusing on the ranks, an approach that is well known in complex systems. Intuitively, a quantile function answers the question: ``Given a certain probability, what is the value below which a specified fraction of the data falls?'' For example, the median is the value at the 50th percentile, meaning half the data lies below it, and the quantile function maps 0.5 to the median value. As we will demonstrate in the next section, quantile functions serve as a valuable tool for modeling deviations from power laws and emergent behavior.

\section{Results}

\subsection{Modulated Pareto distribution yields a BRF}
\label{section_results_1}

Now we state the main result of this work. This result comes from the observation made by Gilchrist \cite{gilchrist2000statistical} that BRF quantile function is the product of two quantile functions (a Pareto and a power distribution) and a result by Sankaran \cite{sankaran2018pareto} on how the interpret the product of two quantile functions. 

\begin{proposition}

Let $X \sim \mbox{Pareto}(x_m,a)$. The random variable $Y$ defined by multiplying $X$ with its cumulative density function to the $b$-th power ($b>0$)
$$
Y = X[F_X(X)]^b
$$
\noindent follows a BRF distribution with parameters $x_m$, $a$ and $b$.
\end{proposition}
\begin{proof}
As seen in the last section, the BRF does not have an analytic solution for neither the pdf nor the cdf, so we will prove $Y$ to follow a BRF by showing that its quantile function is precisely the BRF quantile function. 
Therefore, we aim to compute the quantile function of $Y$. Let $x$ be a realization of $X$ and $y$ the value of $Y$ for which $y=x[F_X(x)]^b$. 
We call $p$ the cumulative probability on these values of $X$ and $Y$,
$p=F_X(x)=F_Y(y)$. We thus have that $y=F^{-1}_Y(p)=Q_Y(p)$, therefore
$Q_Y(p)=x[F_X(x)]^b$. Likewise, $x=F^{-1}_X(p)=Q_X(p)$, so we get
$Q_Y(p)=Q_X(p)[F_X(x)]^b=Q_X(p)p^b$. But we know that $Q_X(p)=\frac{x_m}{(1-p)^a}$, so we finally have the expression for the quantile function of $Y$,
$$
Q_Y(p)=x_m \frac{p^b}{(1-p)^a},
$$
which is the quantile function of a BRF with parameters $x_m, a$ and $b$.
\end{proof}

Note that we can generalize this result by allowing a multiplicative coefficient $c>0$, $Y = cX[F_X(X)]^b$ with $X \sim \mbox{Pareto}(x_m,a)$, then $Y$ is again a BRF random variable with parameters $c\cdot x_m, a$ and $b$.

The straightforward interpretation of this result is as follows: a progressive contraction ``bends'' a power law, which is a straight line in the log-log scale, and transforms it into a DGBD/BRF distribution. This transformation is illustrated graphically in Fig.\ref{fig-prods}-A. We may consider the model effectively drags down the $y$-axis values, and the tail
is down more than the head. When considering normalized observations, this transformation can be understood as a redistribution of resources from the least to the most affluent entities.

Since $0 \leq p \leq 1$, then $y=x\cdot p^b$ is always lower than $x$. This will reduce the quantile values (e.g. the median) for $Y$ at all $p'$s. We can make $Y$ comparable to $X$, for example, by forcing them to have the same median. We can for example define $Y_2=(2p)^b$. The resulting quantile function for $Y_2$ is shown in Fig.\ref{fig-prods}-B in pink. This time, the quantile values for $Y_2$ are lower than that of $X$ when $p \le 0.5$, and higher when $p > 0.5$. This is the situation of $c=2^b$ discussed above. 

Alternatively, we can normalize the original variable $X'=X/\sum X$ (the ranks remain unchanged by \-nor\-ma\-li\-zing), then construct the modulated variable $Y=X'\cdot p^b$ and finally normalize $Y'=Y/\sum Y$. In this way we make $X'$ and $Y'$ comparable as they represent the relative sizes of observations with respect to the total. This normalization allows us to build a generating mechanism for BRF: start with $N$ power law distributed observations, $x_{(r=1)} > x_{(r=2)} > ... > x_{(r=N)}$, normalize ($x'$) and shrink them in such a way that the largest observation remains unchanged, the smallest one reduces reduces to almost zero and the rest is modulated by a power $b>0$ of rank $r$,
$$
\begin{array}{ccl}
y_{(1)} & = &  x'_{(1)},\\
 & & \\
\displaystyle y_{(r)} & = & x'_{(r)} \left(\frac{N+1-r}{N}\right) ^b \approx x'_{(r)} \left(1-\frac{r}{N}\right)^b,\\
 & & \\
\displaystyle y_{(N)} & = & \frac{1}{N} \approx 0 \quad \mbox{when}\quad N>>0,
\end{array}
$$

\noindent and finally normalize $y'=y/\sum y$. The normalized, modulated observations $y$ will follow a BRF distribution. Normalization also enables us to interpret this reduction as a form of redistribution.

\begin{figure}[th]
\begin{center}
\includegraphics[width=0.9\textwidth]{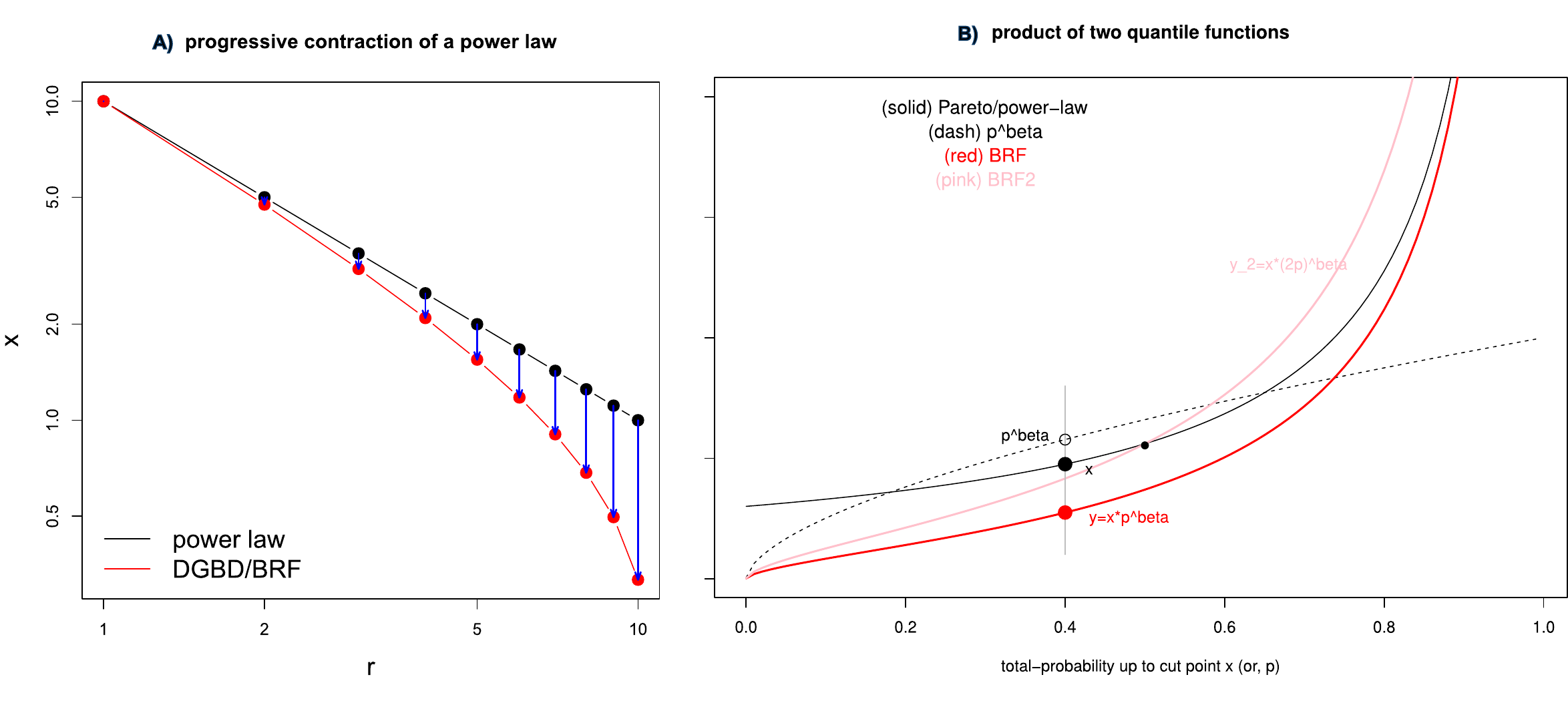}
\end{center}
\vspace{-1cm}
\caption{\textbf{A) Power law contraction.} A progressive contraction transforms a power law into a DGBD/BRF distribution. With proper normalization, this transformation can be interpreted as reallocation of resources from the smallest to the largest entities. \textbf{B) Quantile functions for power distribution, Pareto and BRF}. 
Quantile functions of a Pareto distribution $Q_{pareto}(p)=0.3/(1-p)^{0.9}$
(black solid line), a power-law distribution $Q_{power-law}(p)=p^{0.6}$
(black dashed line), their product $Q_{BRF}(p)=0.3 p^{0.6}/(1-p)^{0.9}$ (red solid line), and another product with an extra coefficient of $2^b$:
$Q_{BRF}(p)=0.3 (2p)^{0.6}/(1-p)^{0.9}$ (pink solid line). 
The quantile values at p=0.4 are marked for the Pareto distribution (black dot), the power-law modified(circle), and their product (red dot). The median (quantile value at p=0.5) is marked for the Pareto distribution and
the second BRF distribution (the two medians are the same).
This figure is modified from the Fig.1.12 in \cite{li1992random}. 
}
\label{fig-prods}
\end{figure}

Our modification process transforms a Pareto-distributed variable into a Beta Rank Function (BRF) variable through a regressive redistribution process, where resources are reallocated from poorer to richer entities. This operation, expressed as $x \rightarrow x \cdot p(x)^b$ with $b > 0$ controlling the redistribution shape, causes a differential contraction that is larger at the lower end and smaller at the higher end of the distribution. Intuitively, this means that starting from power-law observations, reallocating resources upward leads to a BRF-distributed dataset characterized by increased inequality and a breakdown of scale invariance.

\subsection{Numerical simulations}
\label{section_results_2}

To illustrate the generation of a Beta Rank Function (BRF) distributed variable, we created a dataset $X$ consisting of $N=1,000,000$ random observations drawn from a Pareto distribution (a power law) with parameters $x_{\min} = 1$ and shape parameter $a = 1$. These samples can represent quantities such as city sizes, word frequencies, or incomes, examples where power laws and their deviations commonly arise. We then applied the modulation (shrinkage) transformation described in the previous section to this dataset, using three distinct values for the parameter $b$ of the modulating (redistribution) function. This transformation introduces scale-dependent deviations from the original pure power law, with larger deviations occurring among smaller-valued observations. Figure \ref{fig-simulations} presents these results: the top panels display histograms of $\log X$ and the transformed variables $\log Y_i$, plotted on a semilog scale, where $Y_i = X [F_X(X)]^{b_i}$; the middle panels show the corresponding rank-size plots for $X$ and each $Y_i$; and the bottom panels illustrate the shape of the modulating function for each value of $b$.

\begin{figure}[!t]
\centering
\includegraphics[width=0.9\linewidth]{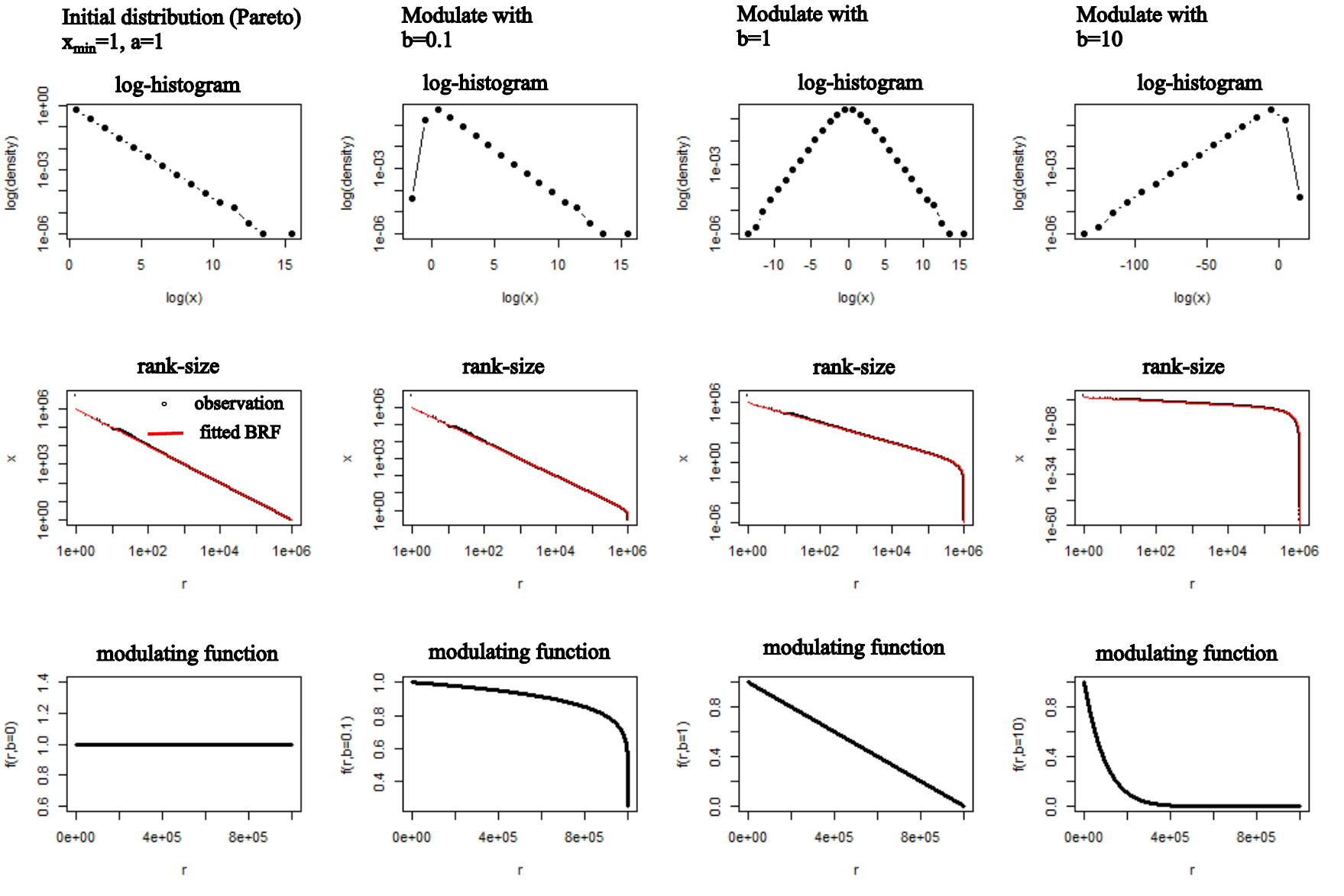}
\caption{\textbf{Effect of power law contraction.} On each column we show the shape of the modulation function with different values of $b$ (lower panels), as well as the resulting distributions, both on log-histogram (upper panels) and rank-size (medium panels) representation. All these are results from stochastic, numerical simulations.}
\label{fig-simulations}
\end{figure}
  
With the initial power law, we see that there is no peak on the histogram (Pareto distribution has a strictly decreasing pdf) and the rank-size function in log scale is a straight line. Then we introduce the transformation with a very low value of the parameter, $b=0.1$. In this case, the modulation is given by a concave function and deviations from the power law are only noticeable for the extreme low-value  observations. This introduces a peak on the histogram with only a few observations on the left side of it, which correspond to the large-rank deviations from power law. A curvature on the rank-size plot is observable only on the far end of the rank $r$. When the modulation parameter is of the same order of magnitude as the power law exponent (both equal to 1 in this case), the observations are evenly distributed on both sides of the mode. Consequently, the log-histogram is symmetric around the peak, and the deviation from the power law is noticeable as a visible elbow on the rank-size plot. Notice that this elbow corresponds to the peak of the histogram or the mode of the distribution. When the modulating parameter is large (in our case, $b=10$), the modulating function has a convex shape. As a result, deviations are large across almost the entire range of scales, yielding a distribution with a long left tail where only a few, very large observations remain unchanged. This causes the rank-size plot to exhibit a more pronounced elbow.

This mechanism of unequally cutting down observation values effectively increases inequality among the sample, because smallest observations get the largest shortening. This is confirmed by measuring the Gini index of our simulated sets, which are equal to 0.87. 0.88, 0.93 and 0.98 for $b=0, 0.1, 1$ and $b=10$ respectively. In this way, an initially unequal  distribution as the power law transforms into an even more unequal distribution. 

\subsection{Example: income distribution}
\label{section_results_3}

In many economies, individual or family income is known to follow a power law rank size function (equivalently, a Pareto distribution), perhaps with the exception near the lowest income end. Depending on the exponent of this power law, this is already a heavy tailed and uneven distribution, yielding the famous 80-20 rule (also called Pareto rule \cite{pareto1919manuale}) where the richest 20$\%$ of the population receive $80\%$ of the total income. But let's assume that, after we have this initial wealth distribution in some system, individuals are subject to a regressive taxation mechanism, where individuals with highest incomes are almost not affected at all, but people on the base of the pyramid lose almost all of their income. 

Countries with regressive tax systems include Norway, Sweden, Denmark, Switzerland and the Netherlands. Factors such as fixed consumption taxes, lack of access to credit, housing market dynamics or healthcare costs  may effectively lead to higher burden on low-income population. Even for countries with progressive tax system, such as USA, richest people may
have access to expert lawyers and accounts to utilize loop holes and end up
paying much less taxes.

According to Proposition 1, if we use the modulation $p^b=((N+1-r)/N)^b$ on a Pareto distributed variable, the resulting distribution will be a BRF distribution with parameters $a$ and $b$, where $a$ is the exponent of the initial power law and $b$ the parameter of the modulating function. The resulting distribution will present a larger inequality than the original one.

We illustrate this process with income data from two different countries: Italy and Mexico. Italy data come from the Nation Survey on Household Income and Wealth for the year 2020 (\url{https://www.bancaditalia.it/statistiche/tematiche/indagini-famiglie-imprese/bilanci-famiglie/}), while Mexico data come from the National Survey of Household Income and Expenditure for the year 2022 (\url{https://en.www.inegi.org.mx/programas/enigh/nc/2022/}). We fit the data to a BRF distribution by means of a non-linear regression on eq. \ref{DGBD} using the Levenberg-Marquardt algorithm \cite{more2006levenberg}. On both cases we observe en excellent agreement between the fitted function and the observed data. 

For the case of Italy, we get $a\approx b =0.47$. The interpretation is the following: there is an initial power law distribution with exponent $a=0.47$, with a Gini index of approximately 0.378. Regressive taxation or other similar mechanisms transform this distribution with a modulation exponent $b$ very similar to the power law exponent; therefore, large deviations from the power law are introduced and the resulting income distribution is even more unequal, with a Gini index of 0.620. 

There are now two characteristic scales, poor people are affected differently than rich people, so behaviour is no longer modeled by a single scaling exponent. In this case, distribution of logarithmic incomes is symmetric around the peak. Notice that the case $a=b$ correspond to the Lavalette distribution, which is very similar to the lognormal distribution but with heavier tales \cite{fontanelli2016beyond}. 

For the case of Mexico, we see $a=0.54$ and $b=0.92$. The original power law has an exponent similar to the Italian case, but now the modulation is more pronounced, producing a left-skewed distribution for logarithmic income, a more pronounced elbow on the rank-size plot and an even greater inequality, reflected on an increase in Gini index from 0.375 to 0.618. The process of power law transformation and the histogram and rank size plots for these two cases are shown on Fig. \ref{fig-ingresos}. 

\subsection{Example: urban area population distribution}
\label{section_results_4}

City population was one the first reported phenomenon to follow a Zipf's law, power law or one of the heavy-tail distributions \cite{auerbach,lotka,stewart,zipf-book,beckmann,berry,richardson,ioannides,soo,jiang}. Two classical explanations for this are 1) that cities grow according to a law of proportionate growth, known as Gibrat’s law \cite{gabaix,saichev}, and 2) an initial log-normal distribution being perturbed by a stochastic Yule process \cite{simon,vitanov}. Deviations from pure power law are often observed,  and there have been debates of whether power law or a lognormal distribution is a better model \cite{fazio}. In our own work, the distribution may also deviate from a Pareto/power-law distribution when the units of population are changed according to administrative needs
\cite{oscar-au}.

We examine relative population of the approximately 2,600 urban areas in the US, according to census data by the Census Bureau for 2020 (\url{https://www.census.gov/programs-surveys/geography/guidance/geo-areas/urban-rural/2020-ua-facts.html}). Again, we observe that BRF/DGBD is a better model than power law, because there is a deviation on the low-population regime. Fitted parameters are $a=1.23$ and $b=0.13$. The value of $a$ is in agreement with the well known Zipf's law, but we observe a small, yet important modulation with exponent $b=0.13$, which manifests in an elbow on the far-right end of the rank-size plot, and also on the emergence of a peak on the logarithmic histogram. This peak introduces two characteristic scales: a power law regime for the majority of urban areas (right side of the peak), but also a different behavior for a few low-populated areas (left side of the peak). 

In terms of the proposed redistribution mechanism, the interpretation is the following: after a proportionate growth process, a power law distribution on urban area population is produced; following this, there is process of unequal population modulation, in which larger cities receive proportionately more resources and grow more than those that receive less; there are selective migration patterns, with people moving from small to larger cities, causing small cities to shrink and large cities to grow, as manifested in deviations from initial power law in the distribution of relative sizes. In our numerical simulations, the initial power law has a Gini coefficient of 0.955, which slightly increases to 0.959 after redistribution.

\begin{figure}[htp]
\centering
\includegraphics[width=0.9\linewidth]{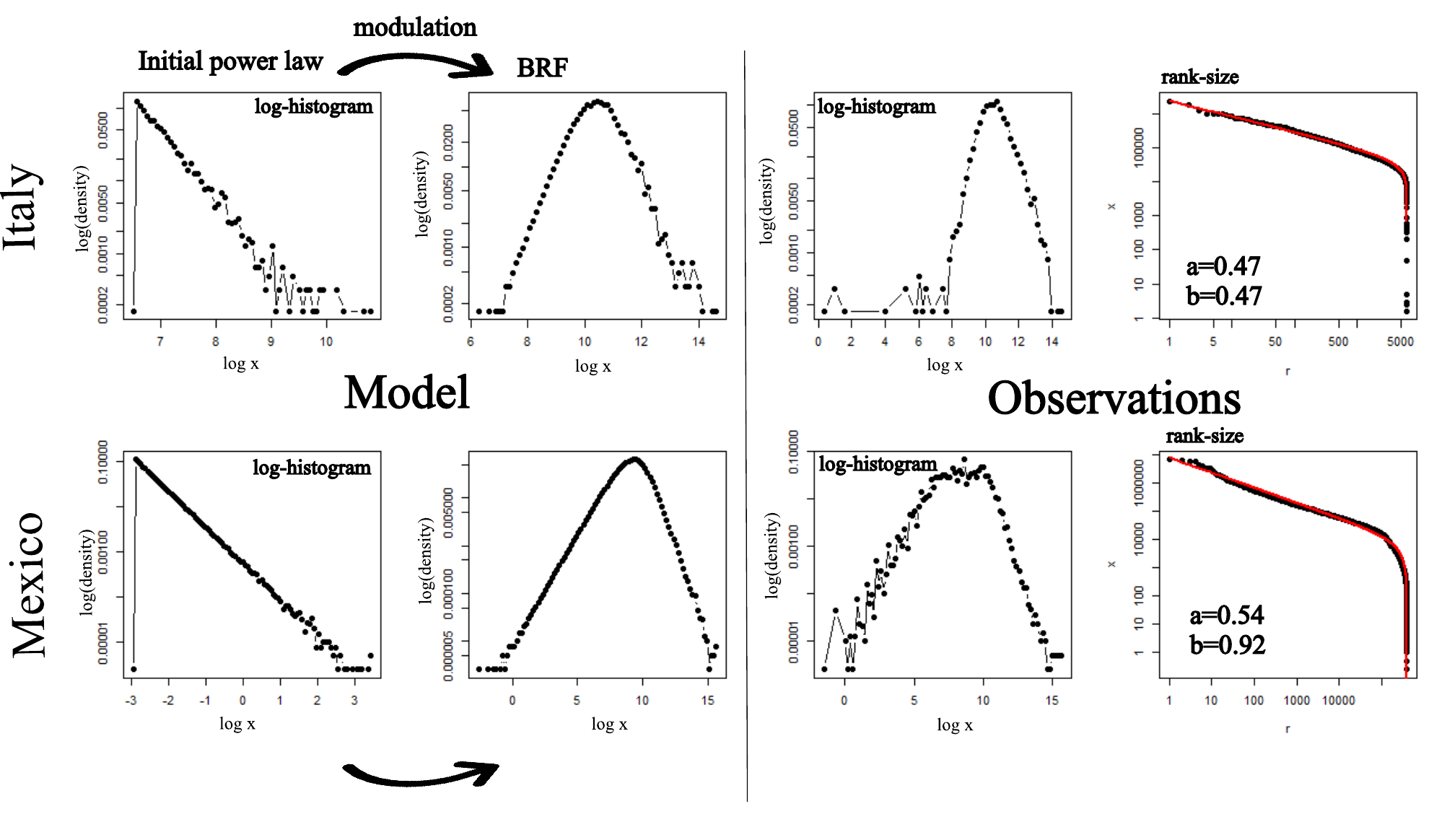}
\caption{\textbf{Modelling the power law transformation in income distribution}. On the left side we show simulated pure-power laws and how they transform into unimodal distributions after modulation or redistribution. On the right side we show the log-histogram of income data for Italy and Mexico, as well as observed rank-size plot and fitted BRF (in red).}
\label{fig-ingresos}
\end{figure} 

\begin{figure}[htp]
\centering
\includegraphics[width=0.9\linewidth]{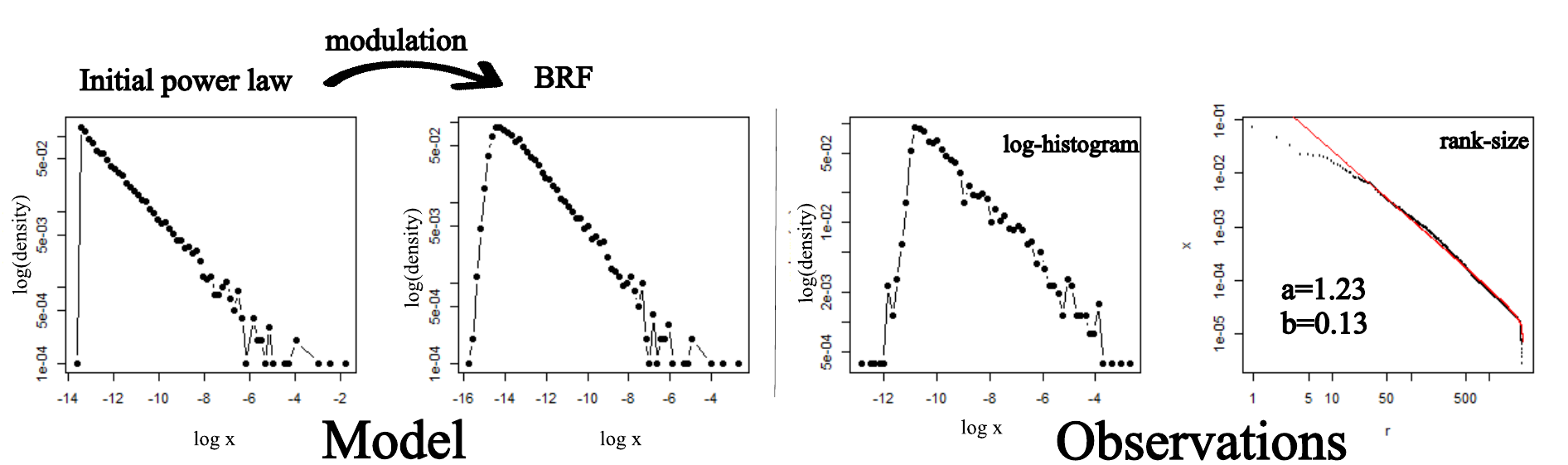}
\caption{\textbf{Modelling the power law transformation in population distribution}. On the left side we show a simulated pure-power law and how it transforms into a unimodal distribution after modulation or redistribution. On the right side we show the log-histogram of urban area population in the US, as well as observed rank-size plot and fitted BRF (in red).}
\label{fig-poblacion}
\end{figure}

These two are examples of phenomena where large observations follow a power law (they are heavy-tailed distributions), but exhibit two characteristic scales when considering all the observations: low-income levels demonstrate different behaviors compared to high-income levels, and sparsely populated areas differ in dynamics from densely populated ones. Each regime is governed by distinct dynamics, a phenomenon frequently observed in complex systems, which is effectively captured by our model and the BRF distribution.

\section{Discussion}

In this work we propose a two-step generation process to produce data that is distributed by BRF (Beta Rank Function) exactly. The first step can be any mechanism that produces a Pareto distributed data (inverse power-law distribution with a lower limit), whereas the second step is a differential modulation or redistribution that increase the inequality (in this sense, it is a ``regressive'' redistribution process). Since BRF has exact analytic expression only for its quantile function, not as pdf or cdf, our generation process, through multiplication in quantile function, seems to be a unique (one and only) way to produce the exact BRF.

We previously proposed a numerical procedure that produces data approximately following BRF: the split-merge model \cite{li2016size,oscar-au}. In the split-merge model, the largest value from the current
population of data is randomly split into two smaller values, and two smallest values combine to form a larger value. When the initial distribution is a Pareto distribution $(a=1, b=0)$, after
certain number of steps, the data will roughly follow a BRF distribution (with the exception of a few largest values). When the initial distribution is a uniform distribution ($a=0, b=1$ \cite{fontanelli2022beta}), the data will also converge to a BRF distribution with the exception of the largest values.

Interestingly, the split-merge model is progressive: it will decrease the inequality, by reducing the largest values and increasing the smallest values. The range of possible values, if starting from a Pareto distribution, will be greatly narrowed in a split-merge model. Split-merge model, in its simplest form, only affects three data points: the largest and the two smallest values. It takes as many steps as the number of data points to propagate changes to all data points. On the other hand, the modification stage in our new two-step model affects all data points in one single sweep.

From Fig.\ref{fig-prods}-A it can be seen that the modulation is carried out vertically to achieve the bending of a straight
line representing the power-law distribution. Another way to bend the straight line is to modulate the points
horizontally. This can be accomplished by a subset sampling, when we discard some samples while keeping others.
We have observed that stopwords follow a DGBD/BRF distribution (Li and Fontanelli, Rank-frequency plots of
stopwords do not follow Zipf’s distribution (2026), manuscript in preparation), and subset sampling may provide a
reasonable explanation of this result.

The BRF is strikingly similar to double Pareto distributions \cite{reed2001pareto, toda2011income}, asymmetric Laplace distribution \cite{kozubowski2001asymmetric} with log-transformed variables \cite{fontanelli2022beta} (skewed log Laplace distribution \cite{dixit2017estimation}), double Pareto lognormal distributions \cite{reed2004double, toda2017note}, etc. All these distributions, when both x- and y-axes are log-transformed, show two straight lines falling from the mode/peak of the distribution on both sides. Unlike some of these distributions, BRF makes a seamlessly smooth transition from one side to another, without a discontinuity in the slope, and without using too many parameters \cite{fontanelli2022beta}. We also note that the name double Pareto in the literature may only refer to the the change of slope value instead of change of the sign of the slope \cite{mitzenmacher2004dynamic}, which is not the situation described by BRF.

The Champernowne model in the field of econometrics uses a Markov transition process between logarithmic scale of income classes to generate either one-tail Pareto or two-tail Pareto distribution, under different conditions \cite{champernowne1953model, toda2012double}. If the transition probabilities are homogeneous, i.e., both upward and downward mobility probabilities are the same for all log-income classes, the resulting equilibrium distribution is Pareto. On the other hand, if the lower income level and upper level classes have different transition probabilities,
the distribution at the two ends can be different. Both Champernowne model and our model are not``mechanical'' models, by which we mean these are not about microscopic mechanisms. They only outline the requirement for achieving double Pareto or BRF distributions. Any detailed mechanism models, as long as they would achieve the required probability, could be
a realization of these outline models.

Unlike the one-sided power law, which is typically used to describe scale-free phenomena, two-sided distributions like the double Pareto or the BRF illustrate phenomena with two distinct characteristic scales. These are empirically observed as two scale-free functions with different exponents, divided by the peak of the distribution. This suggests that phenomena best described by the double Pareto or BRF are governed by different mechanisms operating at different scales. In the case of the BRF, this is reflected in the presence of two shape parameters, $a$ and $b$, which independently control the decay on the right and left sides of the peak, respectively. In the expansion-modification algorithm, the expansion step produces long streaks (associated with persistence), while the modification step produces short streaks (associated with change). These two different scales, large and short streaks, are each dominated by a distinct mechanism within the algorithm and, when fitted to the BRF, are represented by separate exponents: $a$ controls the decay in the large streak regime, while $b$ governs the short streak regime.

Similarly, the split-merge model consists of two distinct dynamics: splitting, which affects only the largest observations, and merging, which impacts the smallest observations. As the process iterates, the effect of splitting propagates through the large observations, while merging propagates through the small ones, eventually leading to a certain level of homogeneity. Once again, different scales (large and small sizes) are governed by different dynamics (splitting and merging). When fitting this process to the BRF, the parameter $b$ is associated with splitting, while $a$ corresponds to merging.

The proposed two-step model captures fundamental dynamics characteristic of complex systems by integrating scale-free generation with scale-dependent redistribution. The initial power law generation reflects universal mechanisms in complex adaptive systems, such as preferential attachment or multiplicative growth, that produce self-organizing, scale-invariant structures affecting all system components uniformly. This underlies the heavy-tailed behavior on the large-size scale, controlled by exponent $a$, which quantifies how extreme events or large elements persist within the system. 

In contrast, the redistribution step introduces feedback-driven, heterogeneous interactions that break pure scale invariance by differentially modulating small and large values. Exponent $b$, governing the redistribution's shape, embodies the system's internal regulatory or feedback processes that disproportionately impact smaller-scale elements, such as resource reallocation, competition, or local constraints. 

Our proposed mechanism and the BRF distribution with its two independent shape parameters together represent the emergence of multi-scale behavior in complex systems, where systemic feedback loops and constraints generate distinct characteristic scales. Together, they model the interplay between endogenous growth dynamics and systemic feedback mechanisms, capturing how rank distributions evolve through both universal scale-free processes and localized redistribution effects. This dual dynamic framework aligns with complexity science themes of emergence, multi-scale interactions, and feedback regulation, providing a mathematical representation of how complex systems balance growth and redistribution to produce observed rank-size patterns with two characteristic scales.

Our numerical simulations demonstrate that the redistribution step breaks the scale-free property of an initial power law, producing a curvature in the log-log rank-size plots that manifests as the characteristic elbow frequently observed in empirical data. This elbow corresponds to a mode or peak in the distribution, showing the presence of two distinct scale-free patterns in two different regime sizes, one on each side of the peak. In complex systems, this reflects the breakdown of pure scale invariance due to multi-scale interactions and feedback mechanisms that differentially affect system components across scales. A subtle downward bend indicates minor deviations from the power law, typically impacting the smaller-scale observations, whereas a pronounced, sharp elbow reveals large and systemic departures from scale-free behavior. Although such large deviations may be less apparent in rank-size plots, they become evident in histogram representations, emphasizing the role of redistribution dynamics in shaping emergent rank distributions and revealing the interplay of endogenous feedback and structural constraints fundamental to complex adaptive systems.

\section{Conclusions}


Power laws are a fundamental feature observed widely across complex systems, reflecting underlying scale invariance and mechanisms such as preferential attachment that drive the emergence of heavy-tailed distributions. These distributions often represent self-organizing processes where no characteristic scale dominates, resulting in straight lines on log-log rank-size plots.

However, empirical data frequently exhibit systematic deviations from ideal power laws, most notably manifesting as elbows or inflection points in log-log rank-size plots. Such elbows correspond to the presence of a mode or peak in the distribution, signaling a departure from pure scale invariance. This peak indicates the existence of two distinct characteristic scales, one governing the behavior on each side of the elbow, thereby breaking the continuous scale symmetry inherent in pure power laws.

The generative mechanism proposed here, combining a power law with a regressive redistribution process, provides a natural explanation for this phenomenon. By modulating the rank distribution through a feedback-driven redistribution function, the model analytically produces the BRF distribution (or its equivalent DGBD rank-size function), which captures this two-scale behavior explicitly. This is also the first exact, analytic derivation of the BRF/DGBD function. 

The BRF/DGBD thus emerges as the exact analytic outcome of this process, offering a rigorous mathematical framework for modeling rank distributions exhibiting concave rank-size plots in the log-log representation (the widely reported elbows). Consistent with extensive empirical literature, the BRF/DGBD fits diverse datasets across multiple domains in complex systems, including socio-economic, linguistic and biological phenomena, with remarkable accuracy.

Importantly, each shape parameter of the BRF controls one characteristic scale or side of the distribution’s peak, providing interpretable links between model parameters and systemic feedback or intrinsic heterogeneity. This dual-parameter structure reflects the interplay of endogenous feedback mechanisms and baseline system tendencies, advancing the understanding of how complex rank distributions evolve dynamically under bounded and feedback-regulated conditions.


The mechanism introduced in this work, together with the BRF distribution, not only deepens conceptual understanding but also opens new avenues for application across systems science. Potential future research could explore its utility in modeling socio-economic inequality, urban dynamics, and the behavior of networked systems, where multi-scale interactions and feedback processes are critical. By bridging a rigorous theoretical process with systemic interpretation, this framework paves the way for cross-disciplinary insights into the complex mechanisms shaping real-world phenomena.

\printbibliography

@article{anderson2006long,
  title={The Long Tail: Why the Future of Business is Selling Less of More},
  author={Anderson, Chris},
  journal={Hyperion google schola},
  volume={3},
  pages={33--50},
  year={2006}
}

@article{mitzenmacher2004brief,
  title={A brief history of generative models for power law and lognormal distributions},
  author={Mitzenmacher, Michael},
  journal={Internet mathematics},
  volume={1},
  number={2},
  pages={226--251},
  year={2004},
  publisher={Taylor \& Francis}
}

@book{sornette2006critical,
  title={Critical phenomena in natural sciences: chaos, fractals, selforganization and disorder: concepts and tools},
  author={Sornette, Didier},
  year={2006},
  publisher={Springer Science \& Business Media}
}

@Inbook{sornette2009,
author="Sornette, Didier",
editor="Meyers, Robert A.",
title="Probability Distributions in Complex Systems",
bookTitle="Encyclopedia of Complexity and Systems Science",
year="2009",
publisher="Springer New York",
address="New York, NY",
pages="7009--7024",
isbn="978-0-387-30440-3",
doi="10.1007/978-0-387-30440-3_418",
url="https://doi.org/10.1007/978-0-387-30440-3_418"
}

@article{lux2016financial,
  title={Financial power laws: Empirical evidence, models, and mechanisms},
  author={Lux, Thomas and Alfarano, Simone},
  journal={Chaos, Solitons \& Fractals},
  volume={88},
  pages={3--18},
  year={2016},
  publisher={Elsevier}
}

@article{li1992random,
  title={Random texts exhibit Zipf's-law-like word frequency distribution},
  author={Li, Wentian},
  journal={IEEE Transactions on information theory},
  volume={38},
  number={6},
  pages={1842--1845},
  year={1992},
  publisher={IEEE}
}

@book{schroeder2009fractals,
  title={Fractals, chaos, power laws: Minutes from an infinite paradise},
  author={Schroeder, Manfred},
  year={2009},
  publisher={Courier Corporation}
}

@article{clauset,
  title={Power-law distributions in empirical data},
  author={Clauset, A and Shalizi, CR and Newman, MEJ},
  journal={SIAM review},
  volume={51},
  number={4},
  pages={661--703},
  year={2009},
  publisher={SIAM}
}

@article{reed2004double,
  title={The double Pareto-lognormal distribution—a new parametric model for size distributions},
  author={Reed, William J and Jorgensen, Murray},
  journal={Communications in Statistics-Theory and Methods},
  volume={33},
  number={8},
  pages={1733--1753},
  year={2004},
  publisher={Taylor \& Francis}
}

@article{gustavo,
  title={Universality of rank-ordering distributions in the arts and sciences},
  author={Mart{\'\i}nez-Mekler, G and Alvarez-Mart{\'\i}nez, R and del R{\'\i}o, MB and Mansilla, R and Miramontes, P and Cocho, G},
  journal={PLoS One},
  volume={4},
  number={3},
  pages={e4791},
  year={2009},
  publisher={Public Library of Science}
}

@article{oscar-au,
  title={Population patterns in World’s administrative units},
  author={Fontanelli, O and Miramontes, P and Cocho, G and Li, W},
  journal={Royal Society open science},
  volume={4},
  number={7},
  pages={170281},
  year={2017},
  publisher={The Royal Society Publishing}
}

@article{ghosh,
  title={Universal city-size distributions through rank ordering},
  author={Ghosh, A and Basu, B},
  journal={Physica A: Statistical Mechanics and its Applications},
  volume={528},
  pages={121094},
  year={2019},
  publisher={Elsevier}
}

@article{wli-physica,
  title={Fitting Chinese syllable-to-character mapping spectrum by the beta rank function},
  author={Li, W},
  journal={Physica A: Statistical Mechanics and its Applications},
  volume={391},
  number={4},
  pages={1515--1518},
  year={2012},
  publisher={Elsevier}
}

@article{wli-jql2,
  title={Characterizing ranked Chinese syllable-to-character mapping spectrum: A bridge between the spoken and written Chinese language},
  author={Li, W},
  journal={Journal of Quantitative Linguistics},
  volume={20},
  number={2},
  pages={153--167},
  year={2013},
  publisher={Taylor \& Francis}
}

@article{wli-jql,
  title={Fitting ranked English and Spanish letter frequency distribution in US and Mexican presidential speeches},
  author={Li, W and Miramontes, P},
  journal={Journal of Quantitative Linguistics},
  volume={18},
  number={4},
  pages={359--380},
  year={2011},
  publisher={Taylor \& Francis}
}

@article{wli-entropy,
  title={Fitting ranked linguistic data with two-parameter functions},
  author={Li, W and Miramontes, P and Cocho, G},
  journal={Entropy},
  volume={12},
  number={7},
  pages={1743--1764},
  year={2010},
  publisher={Molecular Diversity Preservation International}
}

@article{petersen,
  title={Statistical regularities in the rank-citation profile of scientists},
  author={Petersen, AM and Stanley, HE and Succi, S},
  journal={Scientific reports},
  volume={1},
  pages={181},
  year={2011},
  publisher={Nature Publishing Group}
}

@article{fontanelli2020distribuciones,
  title={Distribuciones de probabilidad en las ciencias de la complejidad: una perspectiva contempor{\'a}nea},
  author={Fontanelli, Oscar and Mansilla, Ricardo and Miramontes, Pedro},
  journal={Inter disciplina},
  volume={8},
  number={22},
  pages={11--37},
  year={2020},
  publisher={Universidad Nacional Aut{\'o}noma de M{\'e}xico}
}

@article{fontanelli2016beyond,
  title={Beyond Zipf’s law: the Lavalette rank function and its properties},
  author={Fontanelli, Oscar and Miramontes, Pedro and Yang, Yaning and Cocho, Germinal and Li, Wentian},
  journal={PloS one},
  volume={11},
  number={9},
  pages={e0163241},
  year={2016},
  publisher={Public Library of Science San Francisco, CA USA}
}

@article{del2011general,
  title={General model of subtraction of stochastic variables. Attractor and stability analysis},
  author={del R{\'\i}o, M Beltr{\'a}n and Cocho, G and Mansilla, R},
  journal={Physica A: Statistical Mechanics and its Applications},
  volume={390},
  number={2},
  pages={154--160},
  year={2011},
  publisher={Elsevier}
}

@article{li1991expansion,
  title={Expansion-modification systems: a model for spatial 1/f spectra},
  author={Li, Wentian},
  journal={Physical Review A},
  volume={43},
  number={10},
  pages={5240},
  year={1991},
  publisher={APS}
}

@article{li2016size,
  title={Size distribution of function-based human gene sets and the split--merge model},
  author={Li, Wentian and Fontanelli, Oscar and Miramontes, Pedro},
  journal={Royal Society open science},
  volume={3},
  number={8},
  pages={160275},
  year={2016},
  publisher={The Royal Society}
}

@article{naumis2008tail,
  title={Tail universalities in rank distributions as an algebraic problem: The beta-like function},
  author={Naumis, Gerardo G and Cocho, Germinal},
  journal={Physica A: Statistical Mechanics and its Applications},
  volume={387},
  number={1},
  pages={84--96},
  year={2008},
  publisher={Elsevier}
}

@article{alvarez2014birth,
  title={Birth and death master equation for the evolution of complex networks},
  author={Alvarez-Martinez, R and Cocho, G and Rodr{\'\i}guez, RF and Mart{\'\i}nez-Mekler, G},
  journal={Physica A: Statistical Mechanics and its Applications},
  volume={402},
  pages={198--208},
  year={2014},
  publisher={Elsevier}
}

@article{alvarez2018rank,
  title={Rank ordered beta distributions of nonlinear map symbolic dynamics families with a first-order transition between dynamical regimes},
  author={Alvarez-Martinez, Roberto and Cocho, Germinal and Martinez-Mekler, Gustavo},
  journal={Chaos: An Interdisciplinary Journal of Nonlinear Science},
  volume={28},
  number={7},
  year={2018},
  publisher={AIP Publishing}
}

@article{newman2005power,
  title={Power laws, Pareto distributions and Zipf's law},
  author={Newman, Mark EJ},
  journal={Contemporary physics},
  volume={46},
  number={5},
  pages={323--351},
  year={2005},
  publisher={Taylor \& Francis}
}

@article{fontanelli2022beta,
  title={Beta rank function: A smooth double-Pareto-like distribution},
  author={Fontanelli, Oscar and Miramontes, Pedro and Mansilla, Ricardo and Cocho, Germinal and Li, Wentian},
  journal={Communications in Statistics-Theory and Methods},
  volume={51},
  number={11},
  pages={3645--3668},
  year={2022},
  publisher={Taylor \& Francis}
}

@article{fontanelli2023intermunicipal,
  title={Intermunicipal travel networks of Mexico during the COVID-19 pandemic},
  author={Fontanelli, Oscar and Guzm{\'a}n, Plinio and Meneses-Viveros, Amilcar and Hern{\'a}ndez-Alvarez, Alfredo and Flores-Garrido, Marisol and Olmedo-Alvarez, Gabriela and Hern{\'a}ndez-Rosales, Maribel and Anda-J{\'a}uregui, Guillermo de},
  journal={Scientific Reports},
  volume={13},
  number={1},
  pages={8566},
  year={2023},
  publisher={Nature Publishing Group UK London}
}

@article{sankaran2018pareto,
  title={Pareto-weibull quantile function},
  author={Sankaran, PG and Kumar, DILEEP},
  journal={J Appl Probab},
  volume={13},
  number={1},
  pages={81--95},
  year={2018}
}

@book{gilchrist2000statistical,
  title={Statistical modelling with quantile functions},
  author={Gilchrist, Warren},
  year={2000},
  publisher={CRC Press}
}

@article{hankin,
  title={A new family of non-negative distributions},
  author={Hankin, RKS and Lee, A},
  journal={Aust. N.Z. J. Stat.},
  volume={48},
  number={1},
  pages={67--78},
  year={2006},
}

@book{pareto1919manuale,
  title={Manuale di economia politica con una introduzione alla scienza sociale},
  author={Pareto, Vilfredo},
  volume={13},
  year={1919},
  publisher={Societ{\`a} editrice libraria}
}

@inproceedings{more2006levenberg,
  title={The Levenberg-Marquardt algorithm: implementation and theory},
  author={Mor{\'e}, Jorge J},
  booktitle={Numerical analysis: proceedings of the biennial Conference held at Dundee, June 28--July 1, 1977},
  pages={105--116},
  year={2006},
  organization={Springer}
}

@article{auerbach,
  title={Das gesetz der bev{\"o}lkerungskonzentration [The law of population concentration]},
  author={Auerbach, Felix},
  journal={Petermanns Geographische Mitteilungen},
  volume={59},
  pages={74--76},
  year={1913}
}

@article{beckmann,
  title={City hierarchies and the distribution of city size},
  author={Beckmann, Martin J},
  journal={Economic development and cultural change},
  volume={6},
  number={3},
  pages={243--248},
  year={1958},
  publisher={University of Chicago Press}
}

@article{berry,
  title={City size distributions and economic development},
  author={Berry, Brian JL},
  journal={Economic development and cultural change},
  volume={9},
  number={4, Part 1},
  pages={573--588},
  year={1961},
  publisher={University of Chicago Press}
}

@article{fazio,
  title={Pareto or log-normal? Best fit and truncation in the distribution of all cities},
  author={Fazio, Giorgio and Modica, Marco},
  journal={Journal of Regional Science},
  volume={55},
  number={5},
  pages={736--756},
  year={2015},
  publisher={Wiley Online Library}
}

@article{gabaix,
  title={Zipf's Law and the Growth of Cities},
  author={Gabaix, Xavier},
  journal={American Economic Review},
  volume={89},
  number={2},
  pages={129--132},
  year={1999},
  publisher={American Economic Association}
}

@article{jiang,
  title={Zipf’s law for all the natural cities around the world},
  author={Jiang, Bin and Yin, Junjun and Liu, Qingling},
  journal={International Journal of Geographical Information Science},
  volume={29},
  number={3},
  pages={498--522},
  year={2015},
  publisher={Taylor \& Francis}
}

@article{ioannides,
  title={Zipf’s law for cities: an empirical examination},
  author={Ioannides, Yannis M and Overman, Henry G},
  journal={Regional science and urban economics},
  volume={33},
  number={2},
  pages={127--137},
  year={2003},
  publisher={Elsevier}
}

@article{lotka,
  title={Elements of physical biology},
  author={Lotka, AJ},
  journal={Williams and Wilkins},
  year={1925}
}

@article{richardson,
  title={Theory of the distribution of city sizes: Review and prospects},
  author={Richardson, Harry W},
  journal={Regional Studies},
  volume={7},
  number={3},
  pages={239--251},
  year={1973},
  publisher={Taylor \& Francis}
}

@book{saichev,
  title={Theory of Zipf's law and beyond},
  author={Saichev, Alexander I and Malevergne, Yannick and Sornette, Didier},
  volume={632},
  year={2009},
  publisher={Springer Science \& Business Media}
}

@article{simon,
  title={On a class of skew distribution functions},
  author={Simon, Herbert A},
  journal={Biometrika},
  volume={42},
  number={3/4},
  pages={425--440},
  year={1955},
  publisher={JSTOR}
}

@article{soo,
  title={Zipf's Law for cities: a cross-country investigation},
  author={Soo, Kwok Tong},
  journal={Regional science and urban Economics},
  volume={35},
  number={3},
  pages={239--263},
  year={2005},
  publisher={Elsevier}
}

@article{stewart,
  title={Empirical mathematical rules concerning the distribution and equilibrium of population},
  author={Stewart, John Q},
  journal={Geographical review},
  volume={37},
  number={3},
  pages={461--485},
  year={1947},
  publisher={JSTOR}
}

@article{vitanov,
  title={Test of two hypotheses explaining the size of populations in a system of cities},
  author={Vitanov, Nikolay K and Ausloos, Marcel},
  journal={Journal of Applied Statistics},
  volume={42},
  number={12},
  pages={2686--2693},
  year={2015},
  publisher={Taylor \& Francis}
}

@book{zipf-book,
  title={The Principle of Least Effort},
  author={Zipf, George Kingsley},
  year={1949},
  publisher={CH3}
}

@article{champernowne1953model,
  title={A model of income distribution},
  author={Champernowne, David G},
  journal={The Economic Journal},
  volume={63},
  number={250},
  pages={318--351},
  year={1953},
  publisher={Oxford University Press Oxford, UK}
}

@article{dixit2017estimation,
  title={Estimation of parameters of Skew Log Laplace distribution},
  author={Dixit, VU and Khandeparkar, Pradnya},
  journal={American Journal of Mathematical and Management Sciences},
  volume={36},
  number={4},
  pages={277--291},
  year={2017},
  publisher={Taylor \& Francis}
}

@article{kozubowski2001asymmetric,
  title={Asymmetric Laplace laws and modeling financial data},
  author={Kozubowski, Tomasz J and Podg{\'o}rski, Krzysztof},
  journal={Mathematical and Computer Modelling},
  volume={34},
  number={9-11},
  pages={1003--1021},
  year={2001},
  publisher={Elsevier}
}

@article{mitzenmacher2004dynamic,
  title={Dynamic models for file sizes and double pareto distributions},
  author={Mitzenmacher, Michael},
  journal={Internet Mathematics},
  volume={1},
  number={3},
  pages={305--333},
  year={2004},
  publisher={Taylor \& Francis}
}

@article{reed2001pareto,
  title={The Pareto, Zipf and other power laws},
  author={Reed, William J},
  journal={Economics letters},
  volume={74},
  number={1},
  pages={15--19},
  year={2001},
  publisher={Elsevier}
}

@article{toda2011income,
  title={Income dynamics with a stationary double Pareto distribution},
  author={Toda, Alexis Akira},
  journal={Physical Review E—Statistical, Nonlinear, and Soft Matter Physics},
  volume={83},
  number={4},
  pages={046122},
  year={2011},
  publisher={APS}
}

@article{toda2012double,
  title={The double power law in income distribution: Explanations and evidence},
  author={Toda, Alexis Akira},
  journal={Journal of Economic Behavior \& Organization},
  volume={84},
  number={1},
  pages={364--381},
  year={2012},
  publisher={Elsevier}
}

@article{toda2017note,
  title={A note on the size distribution of consumption: More double Pareto than lognormal},
  author={Toda, Alexis Akira},
  journal={Macroeconomic Dynamics},
  volume={21},
  number={6},
  pages={1508--1518},
  year={2017},
  publisher={Cambridge University Press}
}

\end{document}